\documentclass[11pt,a4paper]{amsart}
\usepackage[margin=25mm]{geometry}
\usepackage[numbers]{natbib}
\usepackage{hyperref}

\usepackage{amssymb,amsmath}
\usepackage{amsthm}
\usepackage{bbm, bm}
\usepackage{stmaryrd}
\usepackage[usenames,dvipsnames]{xcolor}
\usepackage{color}
\usepackage{graphicx}
\usepackage{caption}
\usepackage{float}
\usepackage{wrapfig}

\usepackage{lineno}

\usepackage{abstract}

\input xy
\xyoption{all} 

\numberwithin{equation}{section}

\newtheorem{theorem}{Theorem}[section]

\newtheorem{proposition}[theorem]{Proposition}

\theoremstyle{definition}
\newtheorem{definition}[theorem]{Definition}
\newtheorem{example}[theorem]{Example}

\newenvironment{remark}[1][Remark.]{\begin{trivlist}
\item[\hskip \labelsep {\bfseries #1}]  }{ \end{trivlist}}

\newcommand{\Id}{\mathbbmss{1}}

\DeclareMathOperator{\Hom}{Hom}

\newcommand{\catname}[1]{\textnormal{\texttt{#1}}}

\font\black=cmbx10 \font\sblack=cmbx7 \font\ssblack=cmbx5 \font\blackital=cmmib10  \skewchar\blackital='177
\font\sblackital=cmmib7 \skewchar\sblackital='177 \font\ssblackital=cmmib5 \skewchar\ssblackital='177
\font\sanss=cmss10 \font\ssanss=cmss8 
\font\sssanss=cmss8 scaled 600 \font\blackboard=msbm10 \font\sblackboard=msbm7 \font\ssblackboard=msbm5
\font\caligr=eusm10 \font\scaligr=eusm7 \font\sscaligr=eusm5  \font\fraktur=eufm10
\font\sfraktur=eufm7 \font\ssfraktur=eufm5 
\font\bsymb=cmsy10 scaled\magstep2
\def\all#1{\setbox0=\hbox{\lower1.5pt\hbox{\bsymb
       \char"38}}\setbox1=\hbox{$_{#1}$} \box0\lower2pt\box1\;}
\def\exi#1{\setbox0=\hbox{\lower1.5pt\hbox{\bsymb \char"39}}
       \setbox1=\hbox{$_{#1}$} \box0\lower2pt\box1\;}

\def\tx#1{{\fam0\relax#1}}

\newfam\bifam
\textfont\bifam=\blackital \scriptfont\bifam=\sblackital \scriptscriptfont\bifam=\ssblackital

\newfam\blfam
\textfont\blfam=\black \scriptfont\blfam=\sblack \scriptscriptfont\blfam=\ssblack

\newfam\bbfam
\textfont\bbfam=\blackboard \scriptfont\bbfam=\sblackboard \scriptscriptfont\bbfam=\ssblackboard

\newfam\ssfam
\textfont\ssfam=\sanss \scriptfont\ssfam=\ssanss \scriptscriptfont\ssfam=\sssanss
\def\sss#1{{\fam\ssfam\relax#1}}

\newfam\clfam
\textfont\clfam=\caligr \scriptfont\clfam=\scaligr \scriptscriptfont\clfam=\sscaligr

\newfam\frfam
\textfont\frfam=\fraktur \scriptfont\frfam=\sfraktur \scriptscriptfont\frfam=\ssfraktur

\def\hpb#1{\setbox0=\hbox{${#1}$}
    \copy0 \kern-\wd0 \kern.2pt \box0}
\def\vpb#1{\setbox0=\hbox{${#1}$}
    \copy0 \kern-\wd0 \raise.08pt \box0}

\def\pmb#1{\setbox0\hbox{${#1}$} \copy0 \kern-\wd0 \kern.2pt \box0}
\def\pmbb#1{\setbox0\hbox{${#1}$} \copy0 \kern-\wd0
      \kern.2pt \copy0 \kern-\wd0 \kern.2pt \box0}
\def\pmbbb#1{\setbox0\hbox{${#1}$} \copy0 \kern-\wd0
      \kern.2pt \copy0 \kern-\wd0 \kern.2pt
    \copy0 \kern-\wd0 \kern.2pt \box0}
\def\pmxb#1{\setbox0\hbox{${#1}$} \copy0 \kern-\wd0
      \kern.2pt \copy0 \kern-\wd0 \kern.2pt
      \copy0 \kern-\wd0 \kern.2pt \copy0 \kern-\wd0 \kern.2pt \box0}
\def\pmxbb#1{\setbox0\hbox{${#1}$} \copy0 \kern-\wd0 \kern.2pt
      \copy0 \kern-\wd0 \kern.2pt
      \copy0 \kern-\wd0 \kern.2pt \copy0 \kern-\wd0 \kern.2pt
      \copy0 \kern-\wd0 \kern.2pt \box0}

\mathchardef\za="710B  
\mathchardef\zb="710C  
\mathchardef\zg="710D  
\mathchardef\zd="710E  
\mathchardef\zve="710F 
\mathchardef\zz="7110  
\mathchardef\zh="7111  
\mathchardef\zvy="7112 
\mathchardef\zi="7113  
\mathchardef\zk="7114  
\mathchardef\zl="7115  
\mathchardef\zm="7116  
\mathchardef\zn="7117  
\mathchardef\zx="7118  
\mathchardef\zp="7119  
\mathchardef\zr="711A  
\mathchardef\zs="711B  
\mathchardef\zt="711C  
\mathchardef\zu="711D  
\mathchardef\zvf="711E 
\mathchardef\zq="711F  
\mathchardef\zc="7120  
\mathchardef\zw="7121  
\mathchardef\ze="7122  
\mathchardef\zy="7123  
\mathchardef\zf="7124  
\mathchardef\zvr="7125 
\mathchardef\zvs="7126 
\mathchardef\zf="7127  
\mathchardef\zG="7000  
\mathchardef\zD="7001  
\mathchardef\zY="7002  
\mathchardef\zL="7003  
\mathchardef\zX="7004  
\mathchardef\zP="7005  
\mathchardef\zS="7006  
\mathchardef\zU="7007  
\mathchardef\zF="7008  
\mathchardef\zW="700A  
\mathchardef\zC="7009  

\newcommand{\be}{\begin{equation}}
\newcommand{\ee}{\end{equation}}

\newcommand{\bea}{\begin{eqnarray}}
\newcommand{\eea}{\end{eqnarray}}
\def\*{{\textstyle *}}
\newcommand{\R}{{\mathbb R}}

\newcommand{\Z}{{\mathbb Z}}

\newcommand{\s}{{\textstyle *}}

\def\Hom{\sss{Hom}}

\def\xi{\tx{i}}

\def\s*{{\scriptstyle *}}

\def\cO{\mathcal{O}}

\newcommand{\beas}{\begin{eqnarray*}}
\newcommand{\eeas}{\end{eqnarray*}}

\title{Lie Superheaps and their Groupification} 
\author{Andrew James Bruce }  
   \email{andrewjamesbruce@googlemail.com}

   \date{\today}
 
\begin{document}
 \maketitle
\vspace{-20pt}
\begin{abstract}{\noindent We introduce the notion of a  Lie superheaps as a generalisation of Lie supergroups.  We show that the well-known `groupification' and `heapification' functors generalise to the ambience of supergeometry. In particular, we show that there is an isomorphism between the categories of pointed Lie superheaps and Lie supergroups. To do this, we make extensive use of the functor of points.}\\

\noindent {\Small \textbf{Keywords:} Heaps;~Semiheaps;~Supermanifolds;~ Lie Supergroups;~Generalised Associativity}
\end{abstract}

\section{Introduction and Background}
\subsection{Introduction}
In the 1920s, Pr\"{u}fer \cite{Prufer:1924} and  Baer \cite{Baer:1929} introduced the notion of \emph{heaps} (which are also referred to as torsors, grouds, or herds) as a set with a ternary operation that satisfies a generalisation of associativity known as para-associativity.  A heap can be thought of as a group in which the identity element has been removed. Any group can be turned into a heap by defining the ternary operation as $(x,y,z) \mapsto xy^{-1}z$. Up to isomorphism, every heap arises from a group via this construction. In the other direction, by selecting an element in a heap, one can transform the ternary operation into a group operation, such that the chosen element becomes the identity element. Note, we do not quite have an equivalence of categories, as passing from a heap to a group is not natural. However, there is an isomorphism of categories between pointed heaps and groups.\par 
In an earlier publication (see \cite{Bruce:2024}), the author defined and made an initial study of Lie semiheaps, i.e., semiheaps in the category of smooth manifolds. One of the key results of that paper is the isomorphism of categories between pointed Lie heaps and Lie groups; this extends the algebraic isomorphism to the setting of smooth manifolds (see \cite[Theorem 2.31]{Bruce:2024}). In this note, we extend the notion of Lie semiheaps to the setting of supergeometry. That is, we construct \emph{Lie supersemiheaps} as semiheaps in the category of smooth supermanifolds. The potential difficulty lies in the fact that supermanifolds are not set-theoretical objects. As suggested in \cite{Bruce:2024}, Grothendieck's functor of points is essential in our approach as it allows us to deal with set-theoretical objects and closely mimic the constructions/proofs found in the algebraic setting, albeit at the cost of another level of abstraction. Just as we can define  Lie supergroups as particular group-valued functors, Lie supersemiheaps we define as particular semiheap-valued functors. The main result of this note is that there is an isomorphism between the categories of  Lie supergroups and pointed  Lie superheaps; see Theorem \ref{trm:IsoCats}. \par 
Our main reference for heaps and related structures such as semiheaps is Hollings \& Lawson \cite{Hollings:2017}. For categorical notions, we defer to Mac Lane \cite{MacLane:1988}.  We will freely call on the following works for fundamental results in the theory of supermanifolds \cite{Carmeli:2011,Manin:1997,Varadrajan:2004}.

\subsection{Semiheaps and Heaps}
\begin{definition}
A \emph{semiheap} is a set $S$, possibly empty, equipped with a ternary product
\begin{align*}
\mu :\,& S\times S\times S \longrightarrow S\\
     & (x,y,z) \longmapsto [x,y,z]\,,
\end{align*}
that is para-associative, 
\begin{equation*}
\big[ [x_1,x_2,x_3] , x_4,x_5 \big] = \big [ x_1,[x_4,x_3,x_2],x_5\big] =  \big[ x_1,x_2,[x_3 , x_4,x_5] \big]\, ,
\end{equation*}
for all $x_i \in S$. If $[x,x,y]= [y,x,x] = y$ for all $x,y \in S$, then we speak of a \emph{heap}. A \emph{ semiheap homomorphism} is a map $\phi : S \rightarrow S'$ that respects the ternary product, i.e.,
$$\phi[x,y,z] = [\phi(x), \phi(y), \phi(z)]'\,,$$
for all $x,y,z \in S$. We will denote the resulting categories of semiheaps and heaps as $\catname{SHeap}$ and $\catname{Heap}$. A \emph{pointed (semi)heap} is a (semi)heap with a distinguished point. We will insist that \emph{homomorphisms of pointed (semi)heaps} respect the distinguished points. The resulting category will be denoted $\catname{(S)Heap}_*$.
\end{definition}

\subsection{Supermanifolds and the Functor of Points}
We will assume that the reader has a grasp of the basic theory of supermanifolds. A  \emph{supermanifold} $M := (|M|, \:  \cO_{M})$ of dimension $n|m$, we understand to be a supermanifold as defined by Berezin \& Leites  \cite{Berezin:1976,Leites:1980}, i.e., as a locally superringed space that is locally isomorphic to $\mathbb{R}^{n|m} := \big (\R^{n}, C^{\infty}_{\R^{n}}(-)\otimes \Lambda(\zx^{1}, \cdots \zx^{m}) \big)$. Here, $C^{\infty}_{\R^{n}}(-)$ is the sheaf of smooth real functions on $\R^n$ and $\Lambda(\zx^{1}, \cdots \zx^{m})$ is the Grassmann algebra (over $\R$) with $m$ generators.  The underlying topological space $|M|$ is, in fact, a smooth manifold. This manifold we refer to as the \emph{reduced manifold}. Morphisms of supermanifolds are morphisms as superringed spaces. That is, a morphism $\phi : M \rightarrow N$ consists of a pair $ \phi = (|\phi|, \phi^*  )$, where $|\phi| : |M| \rightarrow |N|$ is a continuous map (in fact, smooth) and  $\phi^*$  is a family of superring homomorphisms $\phi^*_{|V|} : \cO_N(|V|) \rightarrow \cO_M\big( |\phi|^{-1}(|V|)\big)$, for every open $|V| \subset |N|$, that respect the restriction maps. The category of supermanifolds we denote as $\catname{SMan}$. \par
An important and useful result is that 
$$\Hom_{\catname{SMan}}(M, 
N) \cong \Hom_{\catname{SAlg}}(\cO_N(|N|), \cO_M(|M|))\,.$$
where $\catname{SAlg}$ is the category of unital, supercommutative, associative algebras. This allows one to understand morphisms of supermanifolds in terms of homomorphisms between the algebras of global sections.\par 
Grothendieck's functorial approach to algebraic geometry has proven extremely useful in supergeometry and will be used extensively in this note. For any $T = (|T|, \cO_T) \in \catname{SMan}$, the set of \emph{$T$-points of $M$} is defined as 
$$M(T) := \Hom_{\catname{SMan}}(T, M) \cong \Hom_{\catname{SAlg}}(\cO_M(|M|), \cO_T(|T|))\,. $$
That is, we can understand a supermanifold as a contravariant functor $M(-) : \Z_2^n\catname{Man}^{\textrm{op}} \rightarrow \catname{Set}$ from the category of supermanifolds to sets. A fundamental result here is \emph{Yoneda's lemma}, which states that any morphism $\phi :  M \rightarrow N$ defines a natural transformation $\phi_- : M(-) \rightarrow N(-)$ and any natural transformation between $M(-)$ and $N(-)$ arises from a unique morphism between the supermanifolds. Thus, using this categorical point of view captures all the information about supermanifolds and their morphisms. \par 
The Cartesian product of two supermanifolds $M \times N$ is well-defined: the construction of the Cartesian product of manifolds via atlases generalises to supermanifolds.  The universal properties of the Cartesian product of supermanifolds give the useful result that
$$\big( M \times N \big)(T) \cong M(T) \times N(T)\,,$$
which is essential in the definition of a Lie supergroup, and our definition of Lie supersemiheaps and Lie superheaps.\par 
A  Lie supergroup $(G, \mathsf{m}, \mathsf{i} , \mathsf{e})$ can be understood as a supermanifold equipped with three maps, \emph{multiplication, inverse} and \emph{unit},
$$ \mathsf{m} : G \times G \rightarrow G\,, \qquad \mathsf{i}: G \rightarrow G\,, \qquad \mathsf{e} : \R^{0|0}\rightarrow G\,,$$
which satisfy the axioms of a group (see \cite[page 112--113]{Carmeli:2011} for details). A more convenient description is via the functor of points. In particular, we consider a  Lie supergroup as a functor 
$$G(-) : \catname{Sman}^{\textrm{op}} \rightarrow \catname{Grp}\,.$$
\subsection{Pointed Supermanifolds}
We will require pointed supermanifolds. While the notion is clear, we will need to reformulate this to be amenable to the functor of points.
\begin{definition}
A \emph{pointed supermanifold} is a pair $(M, \mathsf{m})$, where $M = (|M|, \cO_M)$ is a supermanifold and $\mathsf{m} :\R^{0|0} \rightarrow M$ is a morphism of supermanifolds.
\end{definition}
Note that the morphism
$$\mathsf{m} :\R^{0|0} =(\{* \}, \R) \longrightarrow M\,,$$
is defined by sending the point $*$ to $m \in |M|$ and $\mathsf{m}^* : \cO_{M,m} \rightarrow \R$ is the evaluation at $m$. This notion of a pointed supermanifold is equivalent to a supermanifold with a distinguished topological point.
\begin{proposition}
If $(M, \mathsf{m})$ is a pointed supermanifold, then $M(T)$ is a pointed set for all $T \in \catname{Sman}$.
\end{proposition} 
\begin{proof}
Recall that $\R^{0|0}$ is a terminal object in $\catname{SMan}$. The unique map $!_T : T \rightarrow \R^{0|0}$ corresponds to the algebra morphism $!_T^\# : \R \rightarrow \cO_T(|T|)$ which sends 
$$\R \ni r\cdot 1 \longmapsto r\cdot \Id_T \in \cO_T(|T|)\,,$$
where $\Id_T$ is the unit function.  Then, $\mathsf{m} \circ !_T  =: \mathsf{m}_T \in M(T)$ gives the a distinguished $T$-point canonically from $\mathsf{m}$. 
\end{proof}
\begin{proposition}
Let $(M, \mathsf{m})$ be a pointed supermanifold. Then under change of parametrisation $\psi \in \Hom_{\catname{SMan}}(T', T)$, the map $\psi^M :M(T)\rightarrow M(T')$ is a map of pointed sets, i.e., $\mathsf{m}_T \mapsto \mathsf{m}_{T'}$.
\end{proposition}
\begin{proof}
The only thing we need to check is that $\mathsf{m}_T \mapsto \mathsf{m}_{T'}$. Under change of parametrisation $\mathsf{m}_T = \mathsf{m}\circ !_T \mapsto \mathsf{m} \circ !_T \circ \psi$. Note that, as we have a terminal object, the map $!_T \circ \psi : T' \longrightarrow \R^{0|0}$ is unique. In terms of global algebra 
\begin{align*}
&\R \stackrel{!_T^\#}{\longrightarrow}  \cO_T(|T|)\stackrel{\psi^\#}{\longrightarrow} \cO_{T'}(|T'|)&\\
& r \cdot 1 \longmapsto r \cdot \Id_T \longrightarrow r\cdot \psi^\#(\Id_T)= r \cdot \Id_{T'}\,,&
\end{align*}
as $\psi^\#$ is a homomorphism of unital algebras.  We have established that $!_T \circ \psi = !_{T'}$. Thus, $\mathsf{m}_T \mapsto \mathsf{m}_{T'}$ under $\psi$.
\end{proof}
\begin{definition}
A \emph{pointed supermanifold morphism} $\phi : (M,\mathsf{m}) \rightarrow (N,\mathsf{n})$  is a morphism of supermanifolds such that $\phi_T : M(T)\rightarrow N(T)$ is a pointed map (preserves the distinguished points) for all $T \in \catname{SMan}$, i.e. $\phi \circ \mathsf{m} =\mathsf{n}$.
\end{definition} 
The discussion above allows us to view pointed supermanifolds as functors
$$ M(-):  \catname{SMan}^{\textrm{op}}\rightarrow \catname{Set}_*\,. $$
%
\section{Lie Supersemiheaps and Related Functors}
\subsection{Lie Supersemiheaps}
We define a Lie supersemiheap as a supermanifold $S$ whose functor of $T$-points takes values in $\catname{SHeap}$. Or, in other words, we define Lie supersemiheaps as functors 
$$S :  \catname{SMan}^{\textrm{op}}\rightarrow \catname{SHeap}\,. $$
More formally; 
\begin{definition}
A supermanifold $S$ is a \emph{Lie supersemiheap} if for any $T \in \catname{SMan}$, the set of $T$-points $S(T)$ is a semiheap, and if for any $\psi : T'  \rightarrow T$ the corresponding map $\psi^S : S(T)\rightarrow S(T')$ is a semiheap homomorphism.  Let $S$ and $S'$ be Lie supersemiheaps. A morphism $\phi : S \rightarrow S'$ is a \emph{Lie supersemiheap homomorphism} if the components of the induced natural transformation $\phi_{-} : S(-)\rightarrow S'(-)$ are all semiheap homomorphisms. The resulting category of  Lie supersemiheaps we will denote as $\catname{LieSuperSHp}$. The full subcategory of super Lie heaps we denote as $\catname{LieSuperHp}$.
\end{definition}
\begin{remark}
In \cite{Bruce:2024}, the categories of Lie semiheaps and Lie heaps were denoted as $\catname{LieSHp}$ and   $\catname{LieHp}$, respectively.  We have inserted ``Super'' rather than just ``S'' in the designation of the `super categories' to avoid notational clashes.
\end{remark}
\begin{example}
The supermanifold $\R^{0|1} = (\{*\}, \Lambda(\theta))$ can be given the structure of a Lie supersemiheap as follows. We can consider global coordinate $\theta$ on $\R^{0|1}$. Then
$$\R^{0 |1}(T) \cong \Hom(\Lambda(\theta), \cO_T(|T|))\,,$$
allows us to consider $T$-points as pullbacks of the coordinate, i.e., as odd elements of the superalgebra  $\cO_T(|T|)$. We then write $\vartheta \in \R^{1|1}(T)$ for a general $T$-point. The semiheap product is 
$$[\vartheta_1, \vartheta_2, \vartheta_3] = \vartheta_1+ \vartheta_2+ \vartheta_3 \,,$$
where the addition is in the algebra $\cO_T(|T|)$ The  change of parametrisation $\psi : T' \rightarrow T$ is completely encoded in the algebra homomorphism $\psi^\# : \cO_T(|T|) \rightarrow \cO_T'(|T'|)$. We thus observe,
$$\psi^\#[\vartheta_1, \vartheta_2, \vartheta_3] = \psi^\#(\vartheta_1)+ \psi^\#(\vartheta_2)+ \psi^\#(\vartheta_3) = [\psi^\#(\vartheta_1), \psi^\#(\vartheta_2),\psi^\#(\vartheta_3)] \,,$$
and so we have defined a Lie semiheap structure on $\R^{0|1}$.
\end{example}
\begin{remark}
In the above example, we have a Lie supersemiheap and not a Lie superheap structure. If we change the sign in the ternary product from `$+++$' to `$+-+$', then we will have a Lie superheap.
\end{remark}
Given a Lie supersemiheap and any $T\in \catname{Sman}$, we have a ternary operation 
$$\mu_T : S^{(3)}(T)  \rightarrow S(T)\,,$$
that defines a semiheap. We have used the shorthand  $S^{(k)} :=  S \times S \times \cdots \times S$, where there are $ k$ factors, and we will use similar notation for maps. We thus have the following commutative diagram that encodes the para-associativity condition (this was first presented in \cite{Bruce:2024}). 
\begin{center}
\leavevmode
\begin{xy}
(0,15)*+{S^{(2)}(T)\times S^{(3)}(T)}="a"; (50,15)*+{S(T) \times S^{(3)}(T)\times S(T)}="b"; (100,15)*+{S^{(3)}(T) \times S^{(2)}(T)}="c";%
(0,-10)*+{S^{(3)}(T)}="d"; (50,-10)*+{S^{(3)}(T)}="e"; (100,-10)*+{S^{(3)}(T)}="f";%
(50, -25)*+{S(T)}="g";
{   \ar@{=} "a";"b"}; {\ar@{=} "b";"c"} ;%
{\ar "a";"d"}?*!/^10mm/{\Id^{(2)}_{S(T)}\times \mu_T }; {\ar "b";"e"}?*!/^22mm/{\Id_{S(T)}\times(\mu_T \circ s_{13})\times \Id_{S(T)}} ; {\ar "c";"f"}?*!/^10mm/{ \mu_T \times \Id^{(2)}_{S(T)}};%
{\ar "d";"g"}?*!/^3mm/{\mu_T};{\ar "e";"g"}?*!/^3mm/{\mu_T}; {\ar "f";"g"}?*!/_3mm/{\mu_T};%
\end{xy}
\end{center}
We understand the map $s_{12}$ as interchanging the first and third copy of $S(T)$, i.e.,
$${\color{red}S(T)}\times S(T) \times {\color{blue}S(T)}\longmapsto {\color{blue}S(T)}\times S(T) \times {\color{red}S(T)}\,.$$
If we have a  Lie superheap, then we have the extra condition
\begin{equation}\label{eqn:HeapCon}
\mu_T \circ (\Id_{S(T)}\times \Delta_T) = \textrm{Prj}_1\,, \qquad \mu_T \circ ( \Delta_T \times \Id_{S(T)}) = \textrm{Prj}_2\,,
\end{equation}
where $\Delta_T : S(T)\rightarrow S(T)\times S(T)$ is the diagonal map and $\textrm{Prj}_i$ is the projection onto the $i$-th factor.\par 
By definition, $\mu_{-}: S^{(3)}(-) \rightarrow S(-)$ is a natural transformation. Thus, Yoneda's lemma says that there exists a corresponding map as supermanifolds $\mu : S^{(3)} \rightarrow S$ that also satisfies the diagram above. If we have a heap, then \eqref{eqn:HeapCon} is also satisfied.  Underlying this is the reduced map $|\mu|: |S|\times |S|\times |S| \rightarrow  |S|$, and this map defines a smooth semiheap structure on the reduced manifold. We thus have the following result. 
\begin{proposition}
Let $S$ be a Lie supersemiheap, then the reduced manifold $|S|$ is a Lie semiheap.   Furthermore, if $S$ is a Lie superheap, then $|S|$ is a Lie heap.
\end{proposition}
We will need the more restrictive notion of a pointed Lie supersemiheap.
\begin{definition}
A \emph{pointed Lie supersemiheap} is a pointed supermanifold $(S, \mathsf{e})$, where $S$ is a Lie supersemiheap.  A \emph{pointed  Lie supersemiheap homomorphism} $\phi :(S,\mathsf{e})\rightarrow (S',\mathsf{e}')$ is a Lie supersemiheap homomorphism that is also a pointed supermanifold morphism.  The resulting category we denote as $ \catname{LieSuperSHp}_*$. The full subcategory of pointed Lie superheaps we denote as $ \catname{LieSuperHeap}_*$.
\end{definition}
\subsection{Heapification and Groupification Functors} From this point on, we will restrict our attention to heaps rather than semiheaps. We will generalise the well-known constructions relating heaps and groups to the setting of supermanifolds. \par 
Consider a Lie supergroup $G\in \catname{LSGrp}$, i.e., $G$ as a functor $G(-) : \catname{SMan}^{\textrm{op}}\rightarrow \catname{Grp}$. For any $T\in \catname{SMan}$, we define
\begin{equation}
[\mathsf{x},\mathsf{y},\mathsf{z}]:= \mathsf{x} \mathsf{y}^{-1} \mathsf{z}\,,
\end{equation}
on the set $G(T)$. It is a classical result that this ternary operation defines a heap. We need to check that any $\psi : T' \rightarrow T$ gives rise to a heap homomorphism $\psi^G : G(T)\rightarrow G(T')$ if we want to construct a  Lie superheap. By definition of a  Lie supergroup, $\psi^G $ is a group homomorphism, specifically $\psi^G(\mathsf{xy})= \psi^G(\mathsf{x})\psi^G(\mathsf{y})$ and $\psi^G (\mathsf{x}^{-1}) = (\psi^G(\mathsf{x}))^{-1}$.  Mapping identity elements to identity elements also holds, but is not part of the construction of a  Lie superheap. Observe that
\begin{align*}
\psi^G[\mathsf{x},\mathsf{y},\mathsf{z}] &= \psi^G(\mathsf{x} \mathsf{y}^{-1} \mathsf{z}) = \psi^G(\mathsf{x}) \psi^G(\mathsf{y}^{-1}) \psi^G(\mathsf{z})  = \psi^G(\mathsf{x}) (\psi^G(\mathsf{y}))^{-1} \psi^G(\mathsf{z})\\ 
&= [\psi^G(\mathsf{x}),\psi^G(\mathsf{y}),\psi^G(\mathsf{z})]\,.
\end{align*}
Thus, $G \in \catname{SLHeap}$. By considering the identity map $\mathsf{e}: \R^{0|0} \rightarrow G$ of the  Lie supergroup, we obtain a pointed  Lie superheap. Note that $\psi^G(\mathsf{e}_T) = \mathsf{e}_{T'}$. \par 
\begin{example}[Translation supergroup]\label{exp:TransGp}
The supermanifold $\R^{1|1} = (\R, C^\infty_{\R}(-)\otimes \Lambda(\theta))$ is a Lie supergroup as follows. We can consider global coordinates $(x, \theta)$ on $\R^{1|1}$. Then
$$\R^{1 |1}(T) \cong \Hom(\cO_{\R^{1|1}}(\R), \cO_T(|T|))\,,$$
allows us to consider $T$-points as pullbacks of the coordinates, i.e., as elements of the superalgebra  $\cO_T(|T|)$. We then write $(\mathsf{x}, \vartheta)\in \R^{1|1}(T)$ for a general $T$-point. The group multiplication is 
$$(\mathsf{x}_1, \vartheta_1)(\mathsf{x}_2, \vartheta_2)= (\mathsf{x}_1 + \mathsf{x}_2 + \vartheta_1 \vartheta_2 , \vartheta_1 + \vartheta_2)\,,$$
and the inverse operation is easily seen to be $(\mathsf{x}, \vartheta) \mapsto  (- \mathsf{x}, - \vartheta)$. The identity is thus $(\mathsf{1}, \mathsf{0})$, where $\mathsf{1}$ is the identity function and $\mathsf{0}$ is the zero function in $\cO_T(|T|)$. Applying the previous construction, the heap structure is
$$[(\mathsf{x}_1, \vartheta_1), (\mathsf{x}_2, \vartheta_2), (\mathsf{x}_3, \vartheta_3)] = (\mathsf{x}_1 - \mathsf{x}_2 +\mathsf{x}_3 - \vartheta_1\vartheta_2 +\vartheta_1\vartheta_3 - \vartheta_2\vartheta_3, \vartheta_1 - \vartheta_2 + \vartheta_3 )\,.$$
\end{example}
\begin{example}[Multiplicative supergroup]\label{exp:MultGp}
The supermanifold $\R_*^{1|1} = (\R_*, C^\infty_{\R_*}(-)\otimes \Lambda(\theta))$, where $\R_* = \R -\{0\}\simeq \textrm{GL}(1)$, is a Lie supergroup as follows. 
Using the set-up of the previous example, we write $(\mathsf{x}, \vartheta)\in \R^{1|1}(T)$ for a general $T$-point. The group multiplication is 
$$(\mathsf{x}_1, \vartheta_1)(\mathsf{x}_2, \vartheta_2)= (\mathsf{x}_1 \mathsf{x}_2 + \vartheta_1 \vartheta_2 , \mathsf{x}_1 \vartheta_2 + \vartheta_1 \mathsf{x}_2)\,.$$
The identity is  $(\mathsf{1}, \mathsf{0})$, where $\mathsf{1}$ is the identity function and $\mathsf{0}$ is the zero function in $\cO_T(|T|)$. The inverse operation is easily seen to be
$$(\mathsf{x}, \vartheta) \mapsto  \left(\frac{1}{\mathsf{x}}, - \frac{1}{\mathsf{x}^2}\vartheta \right )\,. $$ 
Applying the previous construction, the heap structure is
$$[(\mathsf{x}_1, \vartheta_1), (\mathsf{x}_2, \vartheta_2), (\mathsf{x}_3, \vartheta_3)] =\left( \frac{\mathsf{x}_1\mathsf{x}_3}{\mathsf{x}_2} - \frac{ \vartheta_1  \vartheta_2 \mathsf{x}_3}{\mathsf{x}_2^2}
- \frac{  \mathsf{x}_1 \vartheta_2  \vartheta_3}{\mathsf{x}_2^2}  + \frac{\vartheta_1 \vartheta_3}{\mathsf{x}_2}, \frac{\mathsf{x}_1\vartheta_3}{\mathsf{x}_2} - \frac{\vartheta_1 \vartheta_2 \vartheta_3}{\mathsf{x}_2^2} - \frac{\mathsf{x}_1\vartheta_2\mathsf{x}_3}{\mathsf{x}_2^2}+ \frac{\vartheta_1 \mathsf{x}_3}{\mathsf{x}_2}\right)\,.$$
\end{example}
A homorphism of  Lie supergroups $\phi : G \rightarrow G'$ is also a morphism of pointed Lie superheaps. This directly follows from the classical result by considering the group homomorphism $\phi_T :G(T) \rightarrow G'(T)$. 
\begin{definition}
The {heapification functor} is the functor $\mathcal{H}: \catname{LSGrp} \rightarrow \catname{LieSuperHp}_*$ that acts as follows:
\begin{description}
   \item[Objects] $(G, \mathsf{m}, \mathsf{i} ,\mathsf{e}) \longmapsto (G, \mu, \mathsf{e})$, where $[\mathsf{x},\mathsf{y} ,\mathsf{z}] = \mathsf{x}\mathsf{y}^{-1}\mathsf{z}$, for all $\mathsf{x},\mathsf{y},\mathsf{x} \in G(T)$ and all $T \in \catname{SMan}$.
   \item[Morphisms]  $\mathcal{H}(\phi) = \phi$, or all $\phi : G \rightarrow G'$.
\end{description}
\end{definition}
Recall that a functor is full if it is surjective on the hom sets and is faithful if it is injective on the hom sets. A functor is said to be fully faithful if it is full and faithful, i.e., is a bijection between the hom sets. 
\begin{proposition}
The heapification functor is fully faithful.
\end{proposition}
\begin{proof}
Note that the heapification functor acts trivially on a given  Lie supergroup homomorphism; it is seen as living in a different category. Thus, the heapification functor is faithful. The fullness property needs to be examined. We have to check that any pointed  Lie superheap homomorphism $\bar{\phi}: G \rightarrow G'$ is also a Lie supergroup homomorphism. Fixing a $T \in \catname{SMan}$, we observe that
$$\bar{\phi}_T(\mathsf{e}_T) = \mathsf{e}_T'\,,$$
by definition. As a heap homomorphism, it must be the case that
$$ \bar{\phi}_T(\mathsf{x} \mathsf{y}^{-1} \mathsf{z}) = \bar{\phi}_T(\mathsf{x})\big(\bar{\phi}_T(\mathsf{y})\big)^{-1}\bar{\phi}_T(\mathsf{z})\,,$$
for all $\mathsf{x},\mathsf{y},\mathsf{z} \in G(T)$. Firstly, setting $\mathsf{x} =\mathsf{z} = \mathsf{e}_T$   gives 
$$\bar{\phi}_T(\mathsf{y}^{-1}) = \big(\bar{\phi}_T(\mathsf{y})\big)^{-1}\,.$$
Secondly, setting  $\mathsf{y} = \mathsf{e}_T$ gives
$$\bar{\phi}_T(\mathsf{x} \mathsf{z}) = \bar{\phi}_T(\mathsf{x})\bar{\phi}_T(\mathsf{z})\,.$$
Thus, we have a Lie supergroup homomorphism, and so the heapification functor is full. 
\end{proof}
In the other direction, it is well known that given a heap, one can construct a group by selecting a point. Moreover, the selected point is the identity in the constructed group, and groups defined by other choices of the point are isomorphic. We will restrict attention to pointed  Lie superheaps, i.e., the point we consider is privileged. \par 
Consider a pointed  Lie superheap $(S,\mathsf{e})$. For any $T\in \catname{SMan}$, we can define a group on $S(T)$ as
$$\mathsf{xy} :=[\mathsf{x},\mathsf{e}_T,\mathsf{y}]\,, \qquad \mathsf{x}^{-1} := [\mathsf{e}_T,\mathsf{x},\mathsf{e}_T]\,,$$
and the identity element is $\mathsf{e}_T$.  The question is whether we have the right functorial properties to define a Lie supergroup.  That is, given $\psi : T' \rightarrow T$, is the induced map $\psi^{S}:S(T)\rightarrow S(T')$ a group homomorphism?\par 
Under the reparametrisation we have $\psi^S(\mathsf{e}_T) = \mathsf{e}_{T'}$, then using this we see that 
\begin{align*}
& \psi^S(\mathsf{xy}) = \psi^S[\mathsf{x},\mathsf{e}_T,\mathsf{y}] = [\psi^S(\mathsf{x}),\psi^S(\mathsf{e}_T),\psi^S(\mathsf{y})] =  [\psi^S(\mathsf{x}),\mathsf{e}_{T'},\psi^S(\mathsf{y})] =  \psi^S(\mathsf{x})\psi^S(\mathsf{y})\,,\\
& \psi^S(\mathsf{x}^{-1}) = \psi^S [\mathsf{e}_T,\mathsf{x},\mathsf{e}_T] = [\mathsf{e}_{T'},\psi^S(\mathsf{x}),\mathsf{e}_{T'}] = \big(\psi^S(\mathsf{x})\big)^{-1}\,.
\end{align*}
Thus, the above construction produces a  Lie supergroup.
\begin{remark}
We also have a group structure on $S(T)$ for some fixed $T \in \catname{SMan}$ by choosing any element in $S(T)$. The classical construction shows this. However, the arbitrary $T$-point is not inherently tied to a point of $S$ nor a $T'$-point of $S(T')$, and so we do not obtain a functor. 
\end{remark}
Morphisms of pointed  Lie superheaps $\phi :(S,\mathsf{e}) \rightarrow (S', \mathsf{e}')$ preserve the underlying privileged $T$-points, i.e., $\phi_T(\mathsf{e}_T) = \mathsf{e}'_T$. Using this fact, the classical result shows that $\phi_T : S(T)\rightarrow S'(T)$ is a group homomorphism.
\begin{definition}
The {groupification functor} is the functor $\mathcal{G}: \catname{LieSuperHp}_* \rightarrow  \catname{LSGrp}$ that acts as follows:
\begin{description}
   \item[Objects] $  (S, \mu, \mathsf{e})\longmapsto (S, \mathsf{m}, \mathsf{i} , \mathsf{e}) $, where $\mathsf{xy}:= [\mathsf{x}, \mathsf{e}_T,\mathsf{y}]$, $\mathsf{x}^{-1}:= [\mathsf{e}_T, \mathsf{x}, \mathsf{e}_T]$, for all $\mathsf{x},\mathsf{y} \in G(T)$ and all $T \in \catname{SMan}$.
   \item[Morphisms]  $\mathcal{G}(\phi) = \phi$, or all $\phi : S \rightarrow S'$.
\end{description}
\end{definition}
\begin{proposition}
The groupification functor is fully faithful.
\end{proposition}
\begin{proof}
The groupification functor acts trivially on a pointed Lie superheap homomorphism; it is just considered as being in a different category. Thus, the groupification is faithful. The fullness property needs examining. We have to check that any  Lie supergroup homomorphism  $\bar{\phi} : H \rightarrow H'$ is also a pointed  Lie superheap homomorphism. Fixing a $T \in \catname{SMan}$, we observe that
$$\bar{\phi}_T(\mathsf{e}_T) = \mathsf{e}_T'\,,$$
A quick calculation show that $\mathsf{x y}^{-1}\mathsf{z} = [\mathsf{x},\mathsf{y},\mathsf{z}]$ in $H(T)$. Then 
$$ \bar{\phi}([\mathsf{x},\mathsf{y},\mathsf{z}]) = \bar{\phi}(\mathsf{x y}^{-1}\mathsf{z}) = \bar{\phi}(\mathsf{x}) \big(\bar{\phi}(\mathsf{y})\big)^{-1}\bar{\phi}(\mathsf{z}) = [\bar{\phi}(\mathsf{x}),\bar{\phi}(\mathsf{y}),\bar{\phi}(\mathsf{z})]\,,$$
which established the desired result.
\end{proof}
\begin{remark}
If we do not insist on pointed Lie superheaps, then we cannot expect the groupification functor to be fully faithful, as there is no reason why $\bar{\phi}(\mathsf{e}_T) =  \mathsf{e}'_T$ for an arbitrary Lie superheap homomorphism. 
\end{remark}
\begin{example} By setting $(\mathsf{x}_2, \vartheta_2)= (\mathsf{1}, \mathsf{0})$ in the expressions for the heap structure of both Example \ref{exp:TransGp} and Example \ref{exp:MultGp}, the reader can quickly observe that we recover the initial group multiplication.
\end{example}
Two categories $\catname{C}$ and $\catname{D}$ are isomorphic if there exists two functors $\mathrm{F} : \catname{C} \rightarrow \catname{D}$ and $\mathrm{G} : \catname{D} \rightarrow \catname{C}$ such that $\mathrm{F}\mathrm{G} = \Id_D$ and $\mathrm{G}\mathrm{F} = \Id_C$. That is, there is a one-to-one correspondence between objects and morphisms. From the standard algebraic result, it can be seen directly that the heapification and groupification functors are mutual inverses. We have the following theorem.
\begin{theorem} \label{trm:IsoCats}
There is an isomorphism of categories between $\catname{LSGrp}$ and $\catname{LieSuperHp}_*$.
\end{theorem}


\section*{Acknowledgements}  
The author thanks Janusz Grabowski and Rita Fioresi for reading earlier versions of this note.

%

\end{document}